\documentclass[11pt]{amsart}
\usepackage{geometry}                
\usepackage{graphicx}
\usepackage{amssymb}
\usepackage{epstopdf}
\usepackage{verbatim,color,xspace}

\theoremstyle{plain}
\newtheorem{thm}{Theorem}
\newtheorem{lemma}[thm]{Lemma}

\newtheorem*{thm*}{Theorem}
\newtheorem*{lem*}{Lemma}
\newtheorem*{prop*}{Proposition}
\newtheorem*{cor*}{Corollary}
\newtheorem*{conj*}{Conjecture}

\theoremstyle{remark}
\newtheorem{remark}{Remark}

\theoremstyle{definition}
\newtheorem{ex}[thm]{Example}

\newcommand{\diag}{\operatorname{diag}}
\newcommand{\MRCA}{\operatorname{MRCA}}
\newcommand{\tc}{\text{:}} 
\newcommand{\tp}{\text{+}} 

\newcommand{\njst}{$\text{NJ}_{st}$ }
\newcommand{\coalModel}{MSC+$N$\xspace}
\newcommand{\gtrmu}{GTR+$\mu$\xspace}

\newcommand{\dsp}{\displaystyle}

\title[Species trees and log-det distance]{Species tree inference from genomic sequences using the log-det distance}

\author{Elizabeth S. Allman}
\address{PO Box 75666, Department of Mathematics and Statistics, University of Alaska Fairbanks, Fairbanks, AK, 99775}
\email{e.allman@alaska.edu}

\author{Colby Long}
\address{Mathematical Biosciences Institute, The Ohio State University, 1735 Neil Ave., Columbus OH, 43210}
\email{long.1579@mbi.osu.edu}

\author{John A. Rhodes}
\address{PO Box 75666, Department of Mathematics and Statistics, University of Alaska Fairbanks, Fairbanks, AK, 99775}
\email{j.rhodes@alaska.edu}

\date{\today}                                           

\begin{document}
\maketitle

\begin{abstract} 
The log-det distance between two aligned DNA sequences was  introduced as a tool for statistically consistent 
inference of a gene tree under simple non-mixture models of sequence evolution.
Here we prove that the log-det distance, coupled
with a distance-based tree construction method, also permits consistent inference of species trees under 
mixture models appropriate to aligned genomic-scale sequences data.
Data may include sites from many genetic loci, which 
evolved on different gene trees due to incomplete lineage sorting on an ultrametric species tree, with different 
time-reversible substitution processes.
The simplicity and speed of distance-based inference suggests log-det based methods should serve as benchmarks 
for judging more elaborate and computationally-intensive species trees inference methods.
\end{abstract}

\section{Introduction}\label{sec:Intro}

The main result of this work is that statistically consistent inference  of a species tree topology is possible using a very simple log-det distance method, for very general mixture models appropriate to genomic data. Such mixtures not only take into account the coalescent process with fluctuating population sizes, but also allow for variation in base-substitution models across the genome, and certain types of time-dependent heterotachy. While these results require that the species tree is ultrametric when measured in generations, this is plausible for many biological datasets.  Since implementing species tree inference through log-det distance computations and distance-based tree selection is both simple and fast, this method should provide a simple comparison for assessing the behavior of more elaborate methods.

Log-det distances for a phylogenetic mixture model are computed
from weighted sums of the site-pattern frequency matrices associated to 
the component models. 
Thus, the behavior of the distance
depends upon the behavior of the  determinant on sums of matrices.
Although little can be said in general about 
determinants of sums, our approach is to first
investigate properties of the frequency matrices for single-class multispecies coalescent with general time-reversible substitution models. We show these are symmetric, positive definite, 
and thus define positive definite quadratic forms. 
It is properties of these quadratic forms that provide the algebraic means 
necessary to establish our main results.
To our knowledge, this connection to quadratic forms has not been used in phylogenetic work previously. 
This work thus exposes new semi-algebraic aspects of phylogenetics.

\smallskip

To place this work in context,
the increasing availability of genomic-scale  datasets of many aligned gene sequences has made clear that the phylogenetic 
trees inferred for individual genes often differ from one another, and thus cannot be used as proxies for the species tree relating the taxa
as a whole.  Although gene tree discordance might be due to errors in gene tree inference, there are also important biological processes that 
can cause it, which should be taken into account through appropriate modeling.

When the source of gene tree conflict is attributed to incomplete
lineage sorting, the multispecies coalescent (MSC) model \cite{PamiloandNei88,Rannala2003}, combined with standard models of sequence evolution by base substitutions, provides the accepted framework. 
Inference may be performed with a Bayesian approach (e.g., MrBayes/BEST \cite{MrBayes,Liu2007}, *BEAST \cite{Heled2010}), in which a posterior on gene trees and the species tree is computed from sequence data. However,
this approach is  limited in the number of species and genes that can be practically analyzed due to the heavy computational burden. Other methods and software (e.g., STEM \cite{Kubatko2009}, MP-EST \cite{Liu2010}, STAR \cite{Liu2009,ADR2013}, \njst/ASTRID \cite{Liu2011,ASTRID,ADR2018}, ASTRAL-II \cite{ASTRALII}) require less computation, by first inferring individual gene trees, and then treating them as ``data'' for a second inference of a species tree. Although the second stage of this analysis is generally provably consistent, the impact of the error introduced by the first stage is poorly understood.

There are also statistically consistent methods based on the MSC model that avoid inferring individual gene trees at all.
Instead, these methods account for incomplete lineage sorting by viewing concatenated aligned gene sequences as a \emph{coalescent mixture}. (Concatenating gene sequences and analyzing them as if they evolved on a common tree is \emph{not} such a method, and is known to be statistically inconsistent for some parameter regions \cite{kubatko2007inconsistency,roch2015likelihood}.) The coalescent mixture is described by integrating over all possible gene trees, weighted by their probabilities under the MSC. While this leads to a rather complex distribution, it has turned out to be more amenable to analysis than one might naively expect. The first use of this viewpoint, to our knowledge, was in the software SNAPP \cite{RoyChoudhury2008, Bryant2012}, but it also underlies SVDQuartets \cite{CK14} and METAL \cite{Roch15}.

The simplest of these last methods to describe, and the one our results extend, is the METAL approach, in which gene sequences are concatenated, pairwise distances are computed on the alignment, and then a distance method, such as
Neighbor Joining, is used to produce an unrooted species tree. 
As the authors of \cite{Roch15} state, it is ``somewhat surprising" 
that such an approach could be statistically consistently since the standard pairwise distance formulas are derived under a single tree model, and nothing about their form suggests they should apply to mixtures of any kind.
Nonetheless, they provide proofs assuming a certain coalescent mixture of Jukes-Cantor substitution processes, and suggest the Jukes-Cantor assumption can be relaxed. While their focus is on theoretical understanding of the species tree inference problem, and not on practical analysis, a comparative study on
simulated gene datasets  in  \cite{Rusinko17} found the METAL protocol outperformed many other methods. (See also the Discussion in Section \ref{sec:disc} for more comments on appropriate simulation studies.)

\smallskip

Making one seemingly small change to the METAL protocol, in using the log-det distance \cite{Steel94, Lockhart94, Lake94}, we show the resulting inference scheme is a statistically consistent means of inferring the species tree topology under a much broader model than considered in \cite{Roch15}. Our model is much closer, in several aspects, to what might plausibly describe empirical data than any of the models assumed in developing other methods for concatenated genomic-scale data \cite{Bryant2012,CK14,Roch15}. 

The model's most significant new feature is that rather than a base-substitution process that is described by a single general time-reversible (GTR) rate matrix $Q$ across the genome, it allows for an arbitrary mixture with many different $Q$s. This framework  is more in accord with experience analyzing single gene sequences, where empiricists generally consider $Q$ a parameter to be inferred, and find different values for different genes. Even in comparison to the ``somewhat surprising" result of \cite{Roch15}, we believe this result is quite surprising.

Our model also relaxes the assumption made by many methods that the population size and mutation rates on each edge of the species tree remain constant over time. We instead allow population size to vary on each branch of the species tree (as a function of time, measured in generations). 
Finally, for each $Q$, our model allows a changing mutation rate over time, either modeled by a time-dependent scalar-valued rate multiplier, or a more complex process that is also edge-dependent but exhibits a certain symmetry across the genome, giving a relaxation of a molecular clock assumption.
Our strongest results do, however, require that the species tree be ultrametric when edge lengths are given in generations. While this last assumption is not always met in empirical studies, it is one that biological observation can provide  justification for in many instances. If the species tree is not ultrametric in generations, we still extend the results of \cite{Roch15}, but statistical consistency for mixtures with several $Q$s does not hold.

\medskip

This paper is organized as follows:  In Section \ref{sec:GTR-and-LD} we present our basic models and definitions. Section \ref{sec:alg} provides the algebraic lemma on quadratic forms and determinants that underlies our results. In Section \ref{sec:Fab} we obtain properties of the pattern frequency arrays arising from a single rate-matrix coalescent mixture model so that the algebraic lemma can be applied. Section \ref{sec:main} brings the earlier arguments together in our main results for an ultrametric species tree.

In Section \ref{sec:ratesym}, we show our results also apply to an extended model allowing more complicated scalar rate-variation across edges of the tree, provided that rate-variation
satisfies a certain symmetry property. Section \ref{sec:nonultra} considers non-ultrametric species trees, obtaining weaker results and indicating through an example that the full mixture result does not hold. Section \ref{sec:disc} gives concluding remarks.

\section{Definitions and basic lemmas}\label{sec:GTR-and-LD}

In this section we define and establish basic facts about sequence evolution models and coalescent models.  While these involve the standard general time-reversible (GTR) base-substitution models and the multispecies coalescent model (MSC), we emphasize aspects often omitted in species tree inference studies.
In many other works involving these models even if populations sizes or mutation rates are allowed to change, they are assumed to be constant on each edge of the species tree. To broaden this framework, we formalize time-dependent scalar-valued rate variation in the substitution processes, and  give an explicit treatment of changing population sizes in the coalescent process.

There are three time scales inherent in our models. First, on the species tree, time is measured in the natural biological units of \emph{generations}. Second, time in generations and population size (which may vary with time) determine another time scale, in \emph{coalescent units} of generations inversely scaled by population size, which standardizes the rate of the coalescent process. Third, time in generations together with a 
mutation rate (which may vary with time) determine the third time scale, in  \emph{substitution units} of expected number of substitutions-per-site, which standardizes the rate of the mutation process. Since we will also allow the mutation rate to differ for different classes in the mixture, the third scale can vary from class to class. Since researchers in phylogenetics may choose any of these scales as the fundamental one, we emphasize that our models will be formulated explicitly on the first scale, in generations. While the other scales are present, they appear only implicitly in our development.

\subsection{Coalescent process}

Let $(\sigma, \tau)$ denote a rooted, ultrametric species tree,  with topology $\sigma$ and branch lengths $\tau$ measured in number of generations, as in Figure \ref{fig:st}.  Viewing time $t$ from the present, $t=0$, to the past, $t>0$, gene trees are formed within populations on the species tree as modeled by the multispecies coalescent model \cite{PamiloandNei88,Rannala2003}.

\begin{figure}[h]
\centering
\includegraphics[height=2.in]{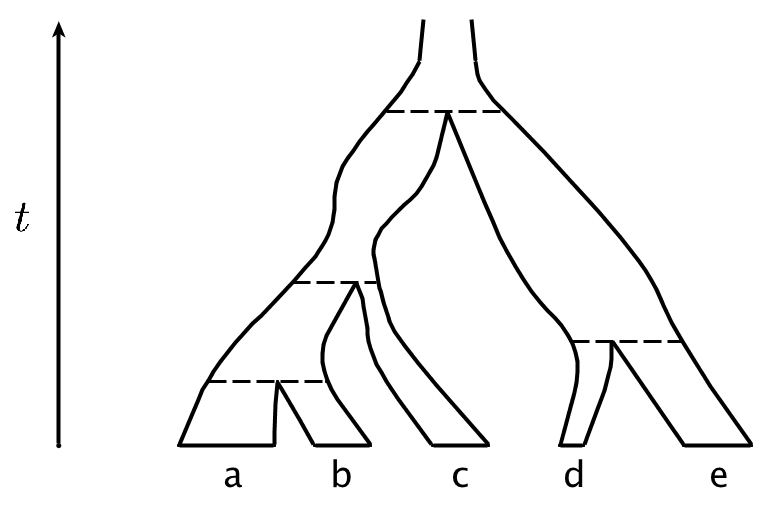}
\caption{A 5-taxon ultrametric species tree with time $t$ in generations before the present, and population functions $N_e(t)$
on each edge depicted using widths of pipes. The edge lengths $\tau$ are measured on the $t$-axis between 
the dotted lines indicating speciation events.}
\label{fig:st}
\end{figure}

On each edge $e$ of the species tree, corresponding to an ancestral population, a time-varying population size is modeled by a function $N_e :[0,\tau_e)\to \mathbb{R}^{>0}$, where $\tau_e$ is the length of $e$ and $N_e(t)$ denotes the population size at time $t$ above the child node on $e$.  For a nominal edge $e$ representing the population ancestral to the root of the tree, of length $\infty$, we also have such a function.
We require that population sizes are bounded above, $N_e(t)\le B$, as this  ensures that with probability 1 lineages eventually coalesce, and is a biological necessity. We also require that $1/N_e(t)$ be integrable over finite intervals.  

The multispecies coalsecent model can be succinctly summarized: If lineages are sampled from populations at the leaves of a species tree then, with time measured in generations into the past, 1) at any time $t\ge 0$, if $K$ lineages are present in a population of size $N(t)$, then the instantaneous rate at which coalescent events occur is $ {\binom K2}/{N(t)}$, and 2) when a coalescent event occurs, 2 of the $K$ lineages are chosen uniformly at random to be the pair that coalesce, so that $K-1$ lineages remain, beginning a new coalescent process, and 3) when lineages reach a node in the species tree, with $K_1$ and $K_2$ lineages from the two child populations, then a new process begins with the combined set of $K_1+K_2$ lineages.

For simplicity, we use the notation \coalModel for this coalescent process with changing population size.

\smallskip

Since our arguments focus on comparing only two sampled lineages at a time, we need an explicit calculation of
the distribution of coalescent times for two lineages in a single population, as given in the following lemma.

\begin{lemma} \label{lem:coalpdf} Suppose two lineages enter a population at $t=0$ with size function $N(t)$,
$0 \le t < \infty$.  Then the time to coalescence has probability density
$$c(t)=c(N;t)= \ell(t) \exp(-L(t)),$$
where $\ell(t) := 1/N(t)$ is the inverse population size, and $L(t) := \int_0^t \ell(\tau) d\tau$ its integral.
\end{lemma}
\begin{proof} 
In the MSC model, 
the instantaneous rate of coalescence for two lineages in a population of size $N$
is $1/N$. Thus if
$p(t)$ denotes the probability that the two lineages have not coalesced by time $t$, then
$$p'(t)=-\frac 1{N(t)}p(t),$$
and therefore 
$$p(t)=\exp \left( -\int_0^t \frac 1{N(\tau)}d\tau \right ).$$
Thus the density function for coalescent times is
$$c(t)=\frac d{dt} (1-p(t))=  \frac 1{N(t)} \exp \left( -\int_0^t \frac 1{N(\tau)}d\tau \right ).$$
\,
\end{proof}

\subsection{Substitution model}\label{subsec:GTR} 

For modeling the evolution of sequences composed of $k$ bases, we use a continuous-time Markov process, with a $k\times k$  instantaneous rate matrix $Q$ such that 1) the off-diagonal entries of $Q$ are non-negative, 2) the rows of $Q$ sum to 0, and 3) $Q$ has stationary distribution $\pi$, with positive entries and $\pi Q=0$. 
The \emph{general time-reversible model} (GTR)  includes the additional assumption 4) $\diag(\pi)Q$ is symmetric. At the root of a gene tree, sites in the ancestral sequence have bases chosen independently with distribution $\pi$.
The Markov matrix $M(\tau)=\exp(Q\tau)$ describes the cumulative substitution process over the time span $\tau$ in the absence of rate variation.  Speed-ups and slow-downs in the substitution process  are introduced through a time-dependent scalar-valued rate function $\mu:\mathbb R^{\ge0}\to \mathbb R^{>0}$,  which we assume to be integrable.  At any time $t$, the substitution process then has instantaneous rate $\mu(t)Q$, with $t$ measured in generations before the present.   We denote this substitution model by \gtrmu.

\begin{lemma}\label{lem:diagQ} Let $Q$ be a GTR rate matrix with stationary distribution $\pi$. Then  $Q=S\Lambda S^{-1}$ where $S=\diag(\pi)^{-1/2} U$ for some  orthogonal matrix $U$, and $\Lambda=\diag(\lambda)$ with $\lambda_1=0$, $\lambda_i\le 0$. If $Q\ne 0$, then $\lambda_i<0$ for some $i$.
\end{lemma}
\begin{proof} The matrix $A=\diag(\pi)^{1/2}Q\diag(\pi)^{-1/2}$ can be written as a convex sum of 
negative-semidefinite matrices of the form $-e_{ii}+e_{ij}+e_{ji}-e_{jj}$, where $e_{ij}$ is the matrix whose only non-zero entry is a $1$ in the $ij$ position. It follows that $A$ is negative-semidefinite and $A=U\Lambda U^T$ for some $U$ orthogonal and $\Lambda$ diagonal with non-positive entries.  Algebra gives the claimed diagonalization of $Q$. Moreover, since rows of $Q$ sum to 0, $Q$ has at least one 0 eigenvalue. If all eigenvalues are 0, then $Q=0.$
\end{proof}

We next characterize the pairwise expected pattern frequencies for the \gtrmu model. For a vector $ v$, we use $\exp(v)$ to denote the entrywise application of the exponential function.

\begin{lemma} \label{lem:Markov} Let 
$Q=S\diag ( \lambda ) S^{-1}$ be the diagonalization of a  GTR rate matrix,   and $\mu(t)$ be a scalar-valued rate function. Then the Markov transition matrix $M(x)= M(\mu,Q,x)$ describing cumulative base substitutions with rate $\mu(t)Q$ for $t\in [0,x]$
is $$M(x)=S\diag(\exp(s(x)\lambda))S^{-1},$$ where $s(x)=\int_0^x \mu(t) dt$. Thus 
the pairwise pattern frequency array 
$F = \diag(\pi)M$ is symmetric positive definite.
\end{lemma}
\begin{proof} Defining 
$M(t)=S\diag(\exp(s(t) \lambda)) S^{-1}$, 
one checks 
\begin{align*}
\frac d{dt} 
M(t)=&\; S\diag(\mu(t) \lambda) \diag(\exp(s(t) \lambda)S^{-1} \\
=&\; \mu(t) (S\diag( \lambda)S^{-1}) (S \diag(\exp(s(t) \lambda)S^{-1}) =\mu(t)QM(t)
\end{align*}
and $M(0)=I$, so $M(x)=M(\mu,Q,x)$. 

By Lemma \ref{lem:diagQ}, we may take $S=\diag(\pi)^{-1/2} U$ for an orthogonal $U$.
If $y \in \mathbb{R}^k$, defining $z^T = y^T\diag(\pi)^{1/2} U$, we have 
$y^TFy = z^T \diag(\exp(s(t) \lambda)) z $. Thus, $F$ is symmetric positive definite.
\end{proof}

\subsection{Coalescent mixture model, and mixtures of coalescent mixtures} \label{subsec:modelM}

By a \emph{coalescent mixture model} on a rooted metric species tree  $(\sigma,\tau)$ on taxa $X$ with population functions $\{N_e\}$ for each edge $e$ of $\sigma$, GTR parameters $Q$ and $\pi$, and scalar-valued rate function $\mu$ we mean the following model of sequence production: For each site, a gene tree is first independently sampled according to the \coalModel model. Then bases for the various taxa in $X$ are sampled according to  the base substitution process \gtrmu on the site's gene tree. This model, with further assumptions that $N_e$ and $\mu$ are constant or edgewise constant, is used in \cite{Bryant2012,CK14,Roch15,LK18}. 

Note that in \cite{LK18} an argument is given that by a reparameterization one can assume $\mu(t)=1$ for all $t$. While that also applies to the coalescent mixture model here, it will not apply to the more general ``mixture of coalescent mixtures'' that will be our main focus, since one cannot in this way reparameterize all mixture classes simultaneously. While we could instead reparameterize to make $N_e(t)=1$ for all $t,e$, that is also problematic for our goals, as it would potentially destroy the ultrametricity of the species tree.

\smallskip

Although the coalescent mixture gives each site its own gene tree, it is easy to see that the same expected site pattern frequencies are produced when this assumption is relaxed to the following model of site non-independence that is more appropriate to genomic data: For multisite genes, suppose the gene length $n$ (in sites) has distribution $f$, independent of all other parameters. Genomic sequence data is then produced for each gene by first sampling a gene length $n$ via $f$, then sampling a gene tree by \coalModel, and finally sampling $n$ sites via \gtrmu. 
While expected pattern frequencies under this gene-independent model and the site-independent model are the same, convergence to the expected values from increasingly large samples will generally be slower for the gene-independent model.

\smallskip

The more general model we study is a \emph{mixture of coalescent mixtures}. We assume that in addition to a species tree and population functions, there is some finite number $m$ of classes, and for each class $1\le i\le m$ parameters $Q_i,\pi_i$ are fixed. In addition there is another parameter, a distribution $ w$, viewed as a non-negative vector whose $m$ entries sum to 1. The data generation process is then that for each site, a class $i$ is first  chosen according to $ w$. Then \coalModel and \gtrmu are used  with parameters $Q_i,\pi_i$  as in the coalescent mixture to generate bases. 

To avoid cumbersome notation, we simply designate this model, which is the main focus of this work, by $\mathcal M$. Its parameters are $(\sigma,\tau)$, $\{N_e\}$, $w$, 
$\{(Q_i,\mu_i)\}_{i=1}^m$. As with the coalescent mixture, one can see that the expected site pattern frequencies of this model are unchanged if it is modified to describe concatenated genes whose lengths are independent of all other parameters. We do not explicitly incorporate this into our model to simplify the presentation. Our results would also extend easily to a model with infinitely-many classes, including a continuum of them, through primarily notational changes to our arguments

\smallskip

Finally,  we consider variants of these models, allowing more general rate variation and non-ultrametric species trees,
but delay presentation of these extensions until Sections \ref{sec:ratesym}  and \ref{sec:nonultra}.

\subsection{Log-det distance}
Given aligned sequence data for two taxa, with $k$ possible bases, the log-det distance between them is defined as follows: 
Let $F_{ab}$ be the $k\times k$ matrix of relative site-pattern frequencies, with the $ij$ entry giving the proportion of sites 
in the sequences exhibiting base $i$ for $a$ and base $j$ for $b$. Let $f_a$ denote the vector of row sums of $F_{ab}$, 
and $f_b$ the vector of column sums, so that these marginalizations are 
simply the proportions of various bases in the sequences of $a$ and $b$. With $g_a$ and $g_b$ the products of the 
entries of $f_a,f_b$, respectively,
\begin{equation}d_{LD}(a,b) =-\frac 1k \left ( \ln |\det(F_{ab})|-\frac 12 \ln (g_ag_b)\right ).\label{eq:dLD}\end{equation}

The log-det distance \cite{Steel94, Lockhart94}, also known as the paralinear distance \cite{Lake94}, was originally developed in the context of the general Markov model, for which it was shown to give an additive function on paths in trees. It is therefore additive on submodels, including the GTR model. For
GTR parameters $Q,\pi$ and time $t$ between the two taxa, the formula for the log-det distance simplifies and
the additive property is apparent:  Specifically, letting $\lambda_1 = 0\ge \lambda_2, \ldots, \lambda_k$ denote the eigenvalues of $Q$, 
$$
d_{LD}(a,b) = -\frac 1k (\lambda_2 + \ldots + \lambda_k)t.
$$

For mixtures of GTR models, in which $F_{ab}$ becomes a sum of many $F_{ab}^i$ for different $Q_i$, we know of 
no such direct way to understand the log-det distance in terms of the parameters, since the determinant does not respect sums.

\subsection{Identifiability and inference}\label{subsec:idinf}
Our results below are phrased in terms of \emph{identifiability} of the species tree topology using either the 3-point or 4-point  
conditions \cite{SempleSteel} applied to log-det distances. This means that given the expected log-det distances 
under the model, one can determine the species tree topology. This result is only a theoretical one, in that the 
expected distances cannot be determined from any  finite data.

However, as the size (number of sites) of a data set produced from the model increases to $\infty$, the empirical 
log-det distances approach the expected ones in probability. Since standard distance-based methods 
(UPGMA for ultrametric trees; NJ, BioNJ, Minimum Evolution, etc.~for any metric trees) are known to correctly 
infer tree topologies in the presence of a sufficiently small amount of noise, this implies that use of them leads to 
statistically consistent inference of the tree topology.

\section{An Algebraic Lemma}\label{sec:alg}

In this section we establish the key lemma enabling an understanding of the log-det distance of equation \eqref{eq:dLD} on mixture models.
In general it is not true that if $\det F_i>\det G_i$ then $\det(\sum_iw_i F_i)>\det (\sum_i w_i G_i)$ , even when all $w_i> 0$. This complicates the analysis of the behavior of the log-det distance, and cast doubts that is might have any good properties for mixtures. However, GTR mixture components $F_i,G_i$, even with changing population size, scalar-valued rates, and a coalescent process, turn out to be associated with positive definite quadratic forms, as we see in Section \ref{sec:Fab}.

\begin{lemma}\label{lem:main}
Suppose for each $i$, $F_i$ and $G_i$ are $k\times k$ symmetric positive definite matrices
 such that $y^TF_iy\ge y^TG_iy$ for every $y\in \mathbb R^k$ with the inequality strict for some $i$ and $y$.
For $w_i> 0$, let $$F=\sum_{i=1}^m w_iF_i,\ \ G=\sum_{i=1}^m w_iG_i.$$
Then $$\det F>\det  G.$$
\end{lemma}

\begin{proof} Since $F_i, G_i$ are positive definite for each $i$, so are $F,G$.

Write $F= UDU^T$ where $D$ is 
diagonal and $U$ is orthogonal. 
Denoting the $i$th column of $U$ by $u_i$, we have
\begin{equation}
\det G = \det(U^TGU)
\leq 
\prod_{i=1}^k (U^TGU)_{ii}\label{eq:Had} =
\prod_{i=1}^k u_i^TGu_i
\end{equation}
where the inequality in \eqref{eq:Had} is Hadamard's for determinants of positive semidefinite matrices \cite{HornJohnson}. But using that
$y^TF_iy\ge y^TG_iy$ for all $y$ yields
\begin{equation}
\prod_{i=1}^k u_i^TGu_i 
\leq \prod_{i=1}^k u_i^TFu_i \label{eq:2nd}
=  \prod_{i=1}^k  u_i^TUDU^Tu_i 
=  \prod_{i=1}^k  D_{ii} = \det F.
\end{equation}
Thus $$\det G\le \det F.$$

In fact, the inequality  is strict, since Hadamard's  inequality \eqref{eq:Had} is strict unless
$U^T G U$ is diagonal. But in that case $F$ and $G$ are simultaneously diagonalized by $U$, and the fact that
$y^TF_iy\ge y^TG_iy$ for all $y$  with a strict inequality for some $i$ and $y$ implies that eigenvalues of $F$ are at least as large as corresponding eigenvalues of $G$, with at least one
strictly larger. This means that the inequality in \eqref{eq:2nd} must be strict.
\end{proof}

\section{Expected pattern frequency matrix under a coalescent mixture}\label{sec:Fab}

In this section we study the expected pattern frequency array for two sequences for a coalescent mixture model. Thus we have only a single \gtrmu process acting in concert with the \coalModel. Our goal is to establish results enabling us to later apply Lemma \ref{lem:main}  to mixtures of these models. Thus we characterize the expected pattern frequency matrices in terms of associated quadratic forms.

Because we consider only two lineages at a time,  these lineages are unable to coalesce until 
some time 
after
$x>0$, where $x$ is the time at which their most recent common ancestor species exists. Since only the population sizes 
of their common ancestors above time $x$ affect the coalescence of the lineages, we can consider a single population 
size function $N(t)$, $t\ge0$, where $N(t)$ for $t< x$ is irrelevant to the process. 

\begin{lemma}\label{lem:delaycoal} Consider a GTR rate matrix $Q\ne 0$, a scalar-valued
rate function $\mu(t)$, and a population function $N(t)$ for $t\ge 0$. 
For any $x\ge0$, let $C(x)=C(Q,\mu,N,x)$ be the expected site-pattern frequency array for two lineages that enter the population at time $0$ and undergo
substitutions 
at rate $\mu(t)Q$,
conditioned on the event that the lineages did not coalesce before time $x$. 

Then for all $0\ne y\in \mathbb R^k$ the function $y^TC(x)y$ is positive and decreasing in $x$, and strictly decreasing for some $y$.
\end{lemma}

\begin{proof}

Let $ N_x(t) = N(x + t)$ and $ \mu_x(t) = \mu(x + t)$, and $M(\mu,Q,x)$ denote the Markov matrix describing the substitution process on a single lineage from time $0$ to $x$ as in Lemma \ref{lem:Markov}. Then, using the 
time-reversibility of the substitution process and the commutativity of the Markov matrices,
\begin{align}C(x)&=\int_0^\infty \diag(\pi) (M(\mu,Q,x)M(\mu_x,Q,z))^2 c_x(z) \, dz\notag\\
&=\diag(\pi) (M(\mu,Q,x))^2\int_0^\infty (M(\mu_x,Q,z))^2 c_x(z) \, dz\label{eq:extra}
\end{align}
where $c_x(z)=\ell_x(z) \exp(-L_x(z))$ is the pdf for coalescent times built from $N_x$, as in Lemma \ref{lem:coalpdf}.

Since by Lemma \ref{lem:Markov} the $M(\mu,Q,w)$ are simultaneously diagonalizable by $S=\diag(\pi)^{-1/2}U$
with diagonalization $\Lambda_{M(\mu,Q,w)}$, we find
$$U^T\diag(\pi)^{-1/2}C(x)\diag(\pi)^{-1/2}U =\Lambda^2_{M(\mu,Q,x)}\int_0^\infty (\Lambda_{M(\mu_x,Q,z)})^2 c_x(z) \, dz. 
$$
Letting $s(t_1,t_2)=\int_{t_1}^{t_2} \mu(t) dt$, it is enough to show that for every eigenvalue $\lambda$ of $Q$, the scalar-valued function 
\begin{align*}
f(\lambda,\mu, N, x)&=\exp(2s(0,x)\lambda)\int_0^\infty \exp(2s(x,x+z)\lambda)c_x(z) dz\\
&= \int_0^\infty \exp(2s(0,x+z)\lambda)c_x(z) dz
\end{align*}
is a decreasing function of $x$, and strictly decreasing for some $\lambda$. 
That is established through Lemma \ref{lem:changepop} and 
Lemma \ref{lem:boundpop} below, together with the fact that $Q$ has at least one negative eigenvalue.
\end{proof}

\begin{remark}\label{rem:extra} The matrix factors in equation \eqref{eq:extra} have a nice interpretation: 
Consider a 2-taxon species tree $(a\tc x,\, b\tc x)$ with branch lengths in generations. If there were 
no coalescent process --- that is, gene lineages came together immediately when entering a population, 
so that all gene trees exactly match the species tree --- then the expected pattern frequency array for 
sequences observed from $a$ and $b$ would be $\diag(\pi) M^2(\mu,Q,x)$, from the two branches 
of length $x$. The coalescent process, however, means gene lineages are delayed in coming together, 
so extra substitutions occur. The Markov matrix
$$\int_0^\infty (M(\mu_x,Q,z))^2 c_x(z)\,dz$$ describes those.
\end{remark}

\begin{remark} If the functions $N$ and $\mu$  are constant, then the matrix $$\int_0^\infty (M(\mu_x,Q,z))^2 c_x(z)\,dz$$ 
is independent of $x$. Thus  for an ultrametric species tree,  pattern frequency arrays for every pair of taxa are a 
particularly simple matrix product, since the second matrix factor encoding  
the extra substitutions due to the coalescent process is the same for all pairs of taxa. 
As a result, appropriately calculated  pairwise distances between taxa 
are simply inflated by a constant from their values in the absence of the coalescent process.
This is a main insight behind the METAL method of \cite{Roch15}, though in that work only the Jukes-Cantor model 
was used, and a much more detailed quantitative study of the distances was undertaken.
\end{remark}

In establishing that the function $f(\lambda,\mu, N, x)$ is decreasing through the following Lemmas \ref{lem:changepop} and 
\ref{lem:boundpop}, the approach used is to view  an increase in $x$ as imposing a further delay in the time to coalescence.
Such a delay can be accomplished by increasing the population size $N(t)$ during the period of the delay, with an increase to infinite population size necessary to prohibit coalescence. 

\begin{lemma} \label{lem:changepop} For $\lambda\le 0$, $x\ge 0$, population function $N$ and
scalar-valued rate function $\mu$, and $s(t_1,t_2)$, $c_x(z)$ as in the proof of Lemma \ref{lem:delaycoal}, let
$$f(\lambda, \mu,N, x)=\int_0^\infty \exp(2s(0,x+z)\lambda)c_x(z) dz.$$
If $\bar N(t)\ge N(t)$ for all $t>0$, and $\bar N(t)> N(t)$ on some subset of $[x,\infty)$ of positive measure and $\lambda<0$ then
$$f(\lambda, \mu,\bar N, x)<f(\lambda, \mu,N, x)$$ for all $x$.

If $\lambda=0$, then $f(\lambda, \mu,N, x)=1.$

\end{lemma}

\begin{proof} Using Lemma \ref{lem:coalpdf}, and
integrating by parts we have
\begin{align*}f(\lambda, \mu,N, x)&=
-\exp(2s(0,x+z)\lambda)\exp(-L_x(z)) |_{z=0}^\infty\\
&\ \ \ \ \ \ \ \ \ \ \ \ \ \ \ \ \ \ \ \ \ \ \ \  +
\int_0^\infty 2\mu(x+z)\lambda \exp(2s(0,x+z)\lambda)\exp(-L_x(z)) dz
 \\
&=\exp(2s(0,x)\lambda)+
\int_0^\infty 2\mu(x+z)\lambda \exp(2s(0,x+z)\lambda)\exp(-L_x(z)) dz
\end{align*}
since the population function $N$ is bounded above so $\lim_{z\to\infty}L(z)= \infty$.

For the case $\lambda<0$, observe $\bar \ell(t)\le  \ell(t)$ so $\bar L_x(t)\le L_x(t)$ for all $t$ while $\bar L_x(t)< L_x(t)$ for all $t$ sufficiently large, so
$$f(\lambda, \mu,\bar N, x)<f(\lambda, \mu,N, x).$$

The claim for $\lambda=0$ is immediate.
\end{proof}

\begin{lemma} \label{lem:boundpop} 
If $\lambda<0$, the function $f(\lambda,\mu,N,x)$ of Lemma \ref{lem:changepop} is a strictly decreasing function of $x$.
\end{lemma}

\begin{proof} Let $0\le x<y$. For any $\dsp u>\max_{t\in [0,y]} N(t)$ define $\bar N^u$ by  
$$\bar N^u(t)=\begin{cases} u&\text{ if $0\le t<y$,}\\N(t)& \text{ if $y\le t$.}\end{cases}$$ 
 
Then
\begin{align*}
f(\lambda,\mu, \bar N^u,x)&=
\int_0^\infty \exp(2s(0,x+z)\lambda) \bar \ell^u_x(z) \exp \left( -\bar L^u_x(z) \right ) dz\\
&>\int_{y-x}^\infty \exp(2s(0,x+z)\lambda) \bar \ell^u_x(z) \exp \left( -\bar L^u_x(z) \right ) dz\\
&=\int_{y-x}^\infty \exp(2s(0,x+z)\lambda) \frac 1{\bar N^u(x+z)} \exp \left( -\int_0^z \frac 1{\bar N^u(x+t)}dt \right ) dz\\
&=\int_{0}^\infty \exp(2s(0,y+z)\lambda) \frac 1{\bar N^u(y+z)} \exp \left( -\int_0^{z+y-x} \frac 1{\bar N^u(x+t)}dt \right ) dz\\
&=\int_{0}^\infty \exp(2s(0,y+z)\lambda) \ell(y+z) \exp \left( -\int_x^{z+y} \frac 1{\bar N^u(t)}dt \right ) dz.\\
\end{align*}
But 
\begin{align*}
\int_x^{y+z} \frac 1{\bar N^u(t)} dt&=\int_x^y \frac 1u dt +\int_y^{y+z} \frac 1{N(t) }dt,\\
&=\frac {y-x}u+\int_y^{y+z} \frac 1{N(t) }dt.
\end{align*}
Thus
\begin{align*}
f(\lambda,\mu, &\bar N^u,x)\\
&> \exp\left (\frac{x-y}{u} \right)\int_0^\infty \exp(2s(0,y+z)\lambda) \ell(y+z) \exp \left( -\int_y^{y+z} \frac 1{N(t)}dt \right )  dz\\
&= \exp((x-y)/u) f(\lambda,  \mu, N, y).
\end{align*}
Thus for any $v>u$, using Lemma \ref{lem:changepop}, we have
$$
f(\lambda,  \mu, N, x)
>f(\lambda,  \mu, \bar N^u, x)> f(\lambda,  \mu, \bar N^v, x)>\exp((x-y)/v) f(\lambda,  \mu, N, y).$$
 Letting $v\to \infty$ gives
$$ f(\lambda,  \mu, N, x)
>f(\lambda,  \mu, \bar N^u, x)\ge  f(\lambda,  \mu, N, y)$$
establishing the claim.
\end{proof}

\section {Main result for ultrametric species trees}\label{sec:main}

We now apply the results of Sections \ref{sec:alg} and \ref{sec:Fab} to the model $\mathcal M$ of Section \ref{subsec:modelM}. 

\begin{thm}\label{thm:main}
Under the mixture of coalescent mixtures model $\mathcal M$  on an ultrametric species tree
measured in numbers of generations, the rooted species tree topology is identifiable from the expected 
log-det distances for pairs of taxa.
\end{thm}
\begin{proof} 
Since a rooted tree is identifiable from its collection of induced rooted triples \cite{SempleSteel}, it will suffice to prove
the result for 3-taxon trees. Without loss of generality, assume $\sigma = ((a,b),c)$. We need only show that
$$d_{LD}(a,b)<d_{LD}(a,c)=d_{LD}(b,c).$$
It is immediate that $d_{LD}(a,c)=d_{LD}(b,c)$, since the model exhibits exchangeability of $a$ and $b$.

Note that under this model, the frequency of bases at any taxon will be $\sum_i w_i\pi_i$ where $\pi_i$ is the base frequency vector for class $i$. In particular
in the formula $$d_{LD}(u,v) =-\frac 14 \left ( \ln |\det(F_{uv})|-\frac 12 \ln (g_ug_v)\right )$$
the last term is identical for all pairs of taxa $u,v$. 

We must show, then, that $\det F_{ab}$, $\det F_{ac} > 0$ and

$$\det F_{ab}>\det F_{ac}.$$

In light of Lemma \ref{lem:main}, it is enough to show that  for a single class $i$, with parameters $\{N_e\}$, $Q_i$, $\mu_i$ that the pairwise frequency arrays $F_i$ and $G_i$ are positive definite, and $y^TF_iy\ge y^TG_iy$ for all $y$, with strict inequality for some $y$. But this is the content of Lemma \ref{lem:delaycoal}.  
\end{proof}

\section{Symmetric rate variation on ultrametric trees}\label{sec:ratesym}
In this section we extend the identifiability result of the previous section, by relaxing the assumption 
that in each mixture component a single rate function $\mu_i(t)$ applies to all simultaneously existing
populations. Instead, for each class $i$ we allow rate functions $\mu^i_e(t)$ to be specified for each edge $e$ of the tree. 
We do, however, require a certain symmetry to the rate functions across classes, so that  in the full model, 
there is a `balancing' of the rates across mixture components. 

\begin{figure}[h]
 \centering
\includegraphics[height=2.in]{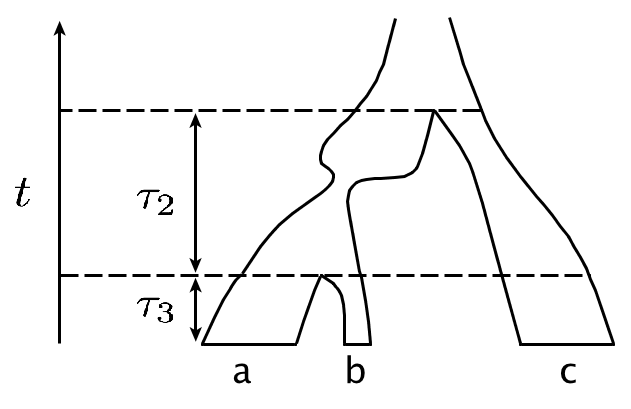}
\caption{The ultrametric species tree $((a\tc\tau_3,b\tc\tau_3)\tc\tau_2, \, c\tc\tau_3+\tau_2)$ with population functions $N_e(t)$ 
depicted using widths of pipes. The dashed lines at speciation times divide the $t$ axis into epochs 3, 2, and 1, 
with $i$ populations coexisting in epoch $i$ of duration $\tau_i$.}\label{fig:symst}
\end{figure}

To be precise, consider an ultrametric 3-taxon species tree with topology $((a,b),c)$ and population functions 
$N_e(t)$ for each edge $e$, as shown in Figure \ref{fig:symst}. With time $t$ measured in generations before the present,
speciation events partition the time axis into three epochs, 
indexed by the count of populations in the tree during the period. In Figure \ref{fig:symst} epoch 3 is $[0,\tau_3)$, 
epoch 2 is $[\tau_3,\tau_3+\tau_2)$, and epoch 1 is $[\tau_3+\tau_2,\infty)$, with epoch $j$ of duration $\tau_j$. 
It is convenient to think of the two epochs in the pendant branch leading to $c$ as corresponding to actual
edges in the tree, so that the species tree has 6 edges or populations in all, including the one above the root.

If scalar-valued rate functions are independently chosen for each of these 6 edges, by permuting 
the three rate functions for epoch 3, and the two for epoch 2, there are twelve assignments of these functions to edges in the tree. 
For fixed GTR model parameters $\pi$, $Q$, we then consider a 12-class equal-weight mixture where gene trees 
for each class are determined  by the multispecies coalescent model, and sequence evolution is modeled using 
$\pi$, $Q$, and one of the 12 equally-likely assignments of rate functions to edges. 
Thus for each class we do not
have a rate function dependent only on $t$ as in the previous section, but in the 12-class mixture 
the mutation process is `on average' the same at all times across populations. 
For instance, in one class, during epoch 3, the substitution process leading to taxon $a$ may be faster than in the population
leading to $b$, but in another class this will be reversed.

\smallskip 
We state a preliminary lemma on pairwise site-pattern frequency arrays for this model.

\begin{lemma}\label{lem:symrate}
Consider the ultrametric species tree $((a\tc\tau_3,b\tc\tau_3)\tc\tau_2,\, c\tc\tau_3\tp\tau_2)$ with branch lengths in generations
and  population functions $\{N_e(t)\}$. Fix  GTR parameters $\pi$, $Q$ and six scalar-valued rate functions  
\begin{align*}
\mu^3_j(t), &\quad t\in [0,\tau_3), \hskip 1.2cm j=1,2,3,\\
\mu^2_i(t),&\quad t\in [\tau_3,\tau_3+\tau_2), \hskip .7cm i=1,2,\\
\mu^1(t),&\quad t\in [\tau_2+\tau_3,\infty).
\end{align*}

Consider a coalescent+base-substitution model  where gene trees arise from the MSC on the species tree with population functions, and base-substitutons occur on these gene trees according to a 12-class uniform mixture  where each class uses the rate function $\mu^1$ above-the-root, a permutation of the rate functions $\mu^2_i$ for the edges in epoch 2, and a permutation of the $\mu^3_j$ for the edges in epoch 3. 

Then the pairwise pattern frequency arrays $F_{ab}$ and $F_{ac}$ arising from this model satisfy
\begin{enumerate}
\item $F_{ab}$, $F_{ac}$ are positive definite,
\item $y^TF_{ab}y \ge y^TF_{ac}y$ for any $y\in \mathbb R^k$, with the inequality strict for some $y$.
\end{enumerate}
\end{lemma}

\begin{proof} With $S_n$ denoting the symmetric group on $\{1,2,\dots, n\}$,
for $\sigma\in S_3$, $\rho\in S_2,$ let $F_{ab}(\sigma, \rho)$ be the pattern frequency array 
from the coalescent+base-substitution model where edges of the species tree are assigned rate functions
as follows:  $\mu^3_{\sigma(1)}$, $\mu^3_{\sigma(2)}$, $\mu^3_{\sigma(3)}$ are assigned to the edges in epoch 3 
from left to right in Figure \ref{fig:symst}, 
$\mu^2_{\rho(1)}$, $\rho^2_{\rho(2)}$ are assigned to edges in epoch 2 from left to right in
Figure \ref{fig:symst}, and $\mu^1$ is assigned to the edge above the root. 
Then 
$$F_{ab}=\sum_{\sigma\in S_3}\sum_{\rho\in S_2} F_{ab}(\sigma, \rho),
\ \ \ \ 
F_{ac}=\sum_{\sigma\in S_3}\sum_{\rho\in S_2} F_{ac}(\sigma, \rho).$$

Let $v=\MRCA(a,b)$,  $r=\MRCA(a,c)$,
$N_v(t)$ be the population function defined piecewise using $N_e(t)$ for edges $e$ above $v$, and $c_v$ the associated 
density of coalescent times.
To compute $F_{ab}(\sigma, \rho)$,  for $i=1,2$
define  rate functions 
$$\tilde \mu^2_i(t)=\begin{cases} \mu^2_{i}(t + \tau_3) &\text{ if } 0 \le t < \tau_2,\\ 
\mu^1(t + \tau_3) &\text{ if } t \ge \tau_2.\end{cases}$$
Then 
\begin{equation*}
F_{ab}(\sigma, \rho)=\int_0^\infty \diag(\pi) M(\mu^3_{\sigma(1)},Q,\tau_3)(M(\tilde \mu^2_{\rho(1)},Q,z))^2 M(\mu^3_{\sigma(2)},Q,\tau_3)c_{v}(z) \, dz.
\end{equation*}
Since the matrix product in the integrand is a pattern frequency array from a GTR coalescent mixture model, by Lemma \ref{lem:Markov} 
it is symmetric positive definite.
As a convex combination of such matrices, $F_{ab}(\sigma,\rho)$ is as well. Claim (1) for $F_{ab}$ follows, 
with a similar argument applying to $F_{ac}$.

To simplify notation, let $M(\mu^j_i)=M(\mu^j_i,Q, \tau_j)$  
denote the Markov matrix describing the base substitution process over 
an edge in epoch $j$ with rate function $\mu^j_i$ and rate matrix $Q$, and 
$M_C(N,\mu)=M_C(N, \mu,Q)$ denote the Markov matrix encoding substitutions arising from the 
coalescent+base substitution model for two lineages entering a population with population function 
$N$, rate function $\mu$, and rate matrix $Q$.  

Since the $M(\mu^j_i)$ and $M_C(N_v,\tilde \mu^2_i)$ are simultaneously diagonalizable, they commute, so

\begin{align}
\diag(\pi)^{-1}F_{ab}&=  \frac1{12} \sum_{\sigma\in S_3}\sum_{\rho\in S_2}  M(\mu^3_{\sigma(1)}) M_C(N_v,\tilde \mu^2_{\rho(1)}) M(\mu^3_{\sigma(2)}) \notag\\
&= \frac2{12} \left (M(\mu^3_1)M(\mu^3_2)+M(\mu^3_1)M(\mu^3_3)+M(\mu^3_2)M(\mu^3_3)\right) \label{eq:Fab}\\
&\ \hskip 2.in \  \times \left ( M_C(N_v,\tilde \mu^2_1)+M_C(N_v,\tilde \mu^2_2)\right)\notag.
\end{align}

Similarly, with $\tilde \mu^1(t)=\mu^1(t+\tau_2+\tau_3)$,
\begin{align}
\diag(\pi)^{-1}F_{ac}&=\frac1{12} \sum_{\sigma\in S_3}\sum_{\rho\in S_2} M(\mu^3_{\sigma(1)}) M(\mu^2_{\rho(1)})M_C(N_r, \tilde \mu^1)M(\mu^2_{\rho(2)}) M(\mu^3_{\sigma(3)})\notag \\
&= \frac4{12} \left (M(\mu^3_1)M(\mu^3_2)+M(\mu^3_1)M(\mu^3_3)+M(\mu^3_2)M(\mu^3_3)\right)\label{eq:Fac} \\
&\ \hskip 2.in \  M(\mu^2_1)M(\mu^2_2)  M_C(N_r, \tilde \mu^1).\notag
\end{align}

Now Lemma \ref{lem:Markov}  implies all matrices $M(\cdot)$ and $M_C(\cdot)$ in Equations \eqref{eq:Fab} and \eqref{eq:Fac} are simultaneously
diagonalizable by $S=\diag(\pi)^{-1/2}U$ for some orthogonal $U$. Thus there are diagonalizations 
$$\diag(\pi)^{-1}F_{ab}=S\Lambda_{ab}S^{-1}, \ \ \ \ \diag(\pi)^{-1}F_{ac}=S\Lambda_{ac}S^{-1},$$
so with $R=\diag(\pi)^{1/2}U$,
$$F_{ab}=R\Lambda_{ab}R^T,\ \ \ F_{ac}=R\Lambda_{ac}R^T.$$
To establish claim (2), it is therefore enough to show that  eigenvalues of $\diag(\pi)^{-1}F_{ab}$ are greater than or equal to the corresponding 
ones for $\diag(\pi)^{-1}F_{ac}$, with at least one strictly greater.
 
To facilitate our argument, we remove the common factor in Equations \eqref{eq:Fab} and \eqref{eq:Fac} corresponding to epoch 3, as well as a scalar $2/12$,
since this factor is diagonalizable by $S$ and has positive eigenvalues.
It is thus sufficient to show that for
$$\tilde F_{ab}=  M_C(N_v,\tilde \mu^2_1)+M_C(N_v,\tilde \mu^2_2) $$ and
$$\tilde F_{ac}= 2M(\mu^2_1)M(\mu^2_2)  M_C(N_r, \mu^1),$$
eigenvalues of $\tilde F_{ab}$ are larger than those of $\tilde F_{ac}$, with at least one strictly larger.

By Lemma \ref{lem:delaycoal}, for any $y\in \mathbb R^k$, $i=1,2$,
$$y^T\diag(\pi) M_C(N_v,\tilde \mu^2_i)y\ge y^T\diag(\pi) M(\mu^2_i)^2M_C(N_r,\mu^1)y,$$ with a strict inequality for some $y$.
But again using the relationship between the  diagonalizations of the Markov matrices and of the quadratic forms, this implies that the eigenvalues of $M_C(N_v,\tilde \mu^2_i)$ are larger than those of $M(\mu^2_i)^2M_C(N_r,\mu^1)$ with at least one strictly larger.

Summing over $i=1,2$
gives the same relationship between eigenvalues of $\tilde F_{ab}$ and $\left (M(\mu_1^2)^2+M(\mu_2^2)^2\right )M_C(N_r,\mu^1)$.
Using the scalar inequality $\alpha^2+\beta^2\ge 2\alpha \beta$ on corresponding eigenvalues of $M(\mu_1^2)$ and $M(\mu_2^2)$ gives the desired relationships on eigenvalues of $\tilde F_{ab}$ and $\tilde F_{ac}$.
\end{proof}

\smallskip

 We now define a more general coalescent+base substitution mixture model on an $n$-taxon ultrametric tree denoted by $\mathcal M^{sym}$: 
For each site, a gene tree is determined by the MSC on an ultrametric species tree with populations functions $\{N_e\}$. With edges 
subdivided into epochs determined by speciation times, suppose for each class $i$ with weight $w_i$ there are GTR parameters 
$\pi_i,Q_i$, and scalar-valued rate functions $\mu_i(e,t)$ for edges $e$ within epochs. Assume in addition the following symmetry 
condition:  
for any class $i$ and any permutation of the $\mu_i(e,t)$ within epochs, there is another class $j$ with the same 
$w_i,\pi_i,Q_i$ but the permuted rate functions.

A model in $M^{sym}$ on a large tree restricts to a model in $\mathcal M^{sym}$ on any induced 3-taxon subtree. This 
3-taxon model is easily seen to be a mixture of models of the type in Lemma \ref{lem:symrate}. Repeating the 
argument of Theorem \ref{thm:main}, but using Lemma \ref{lem:symrate} in place of Lemma \ref{lem:delaycoal}, 
yields the following.

\begin{thm}\label{thm:sym}
Consider a model in $\mathcal M^{sym}$ on an ultrametric tree. Then the rooted species tree topology is identifiable from the 
expected log-det distances between taxa.
\end{thm}

\section{Non-ultrametric trees and edge dependent rate variation}\label{sec:nonultra}

Having so far only considered species trees that are ultrametric in units of generations, we turn now to non-ultrametric species trees.
We will show that the log-det distance is poorly behaved on a model with a mixture of coalescent mixtures. Nonetheless, for a single-class coalescent mixture model we can still obtain strong results.
 
 \medskip
 
To explore the behavior of the log-det distance for non-ultrametric trees, we focus on the 4-point condition 
\cite{SempleSteel} to determine if pairwise distances fit an unrooted tree.
The $4$-point condition states that for each unrooted quartet tree $wx|yz$ displayed on the species tree
\begin{equation}
d(w,x)+d(y,z)< d(w,y)+d(x,z)=d(w,z)+d(x,y).\label{eq:4point}
\end{equation}
As the next example shows, for mixtures of coalescent mixtures on non-ultrametric species tree,
Inequality and Equation \eqref{eq:4point} need not hold.

\begin{ex} \label{ex:failure_ex}
Consider an MSC model on the $4$-leaf non-ultrametric rooted species 
tree $(\sigma,\tau) = ((a\tc N, b\tc N/10)\tc N/20,(c\tc N/10,d\tc N)\tc N/20 )$ with
branch lengths in numbers of generations and $N$ a constant population size over the entire tree.
Let $Q$ be a Jukes--Cantor rate matrix with off-diagonal entries
equal to $.01/N$, and suppose that for a 2-class mixture the
substitution process on each gene tree is chosen to be either $Q$ or $20Q$ with equal probability.

Then the expected site-pattern probability distribution can be computed exactly, giving, to several digits of accuracy,
\begin{align*}
d_{LD}(a,c)  + d_{LD}(b,d) \approx 0.996,  \\
d_{LD}(a,d)  + d_{LD}(b,c) \approx 0.952, \\
d_{LD}(a,b)  + d_{LD}(c,d) \approx 0.973.
\end{align*}

Since $d_{LD}(a,c)  + d_{LD}(b,d) \neq d_{LD}(a,d)  + d_{LD}(b,c)$, the 4-point equality fails. Moreover,
$d_{LD}(a,b)  + d_{LD}(b,c)$ fails to be the smallest of the three sums.

Indeed, the failure of the 4-point condition in this example is not linked to the inclusion of the coalescent process 
in the model.  
For a simpler mixture, in which sequences are generated on a $2$-class mixture on the metric gene tree
$(\sigma,\tau)$, then using the mutation parameters given above
we find
\begin{align*}
d_{LD}(a,c)  + d_{LD}(b,d) \approx 0.605,  \\
d_{LD}(a,d)  + d_{LD}(b,c) \approx  0.534, \\
d_{LD}(a,b)  + d_{LD}(c,d) \approx 0.566. 
\end{align*}
That the two classes in this example differ only by a rescaling of substitution rates suggests it would be hard to find 
any general circumstances on which log-det distances from mixtures on non-ultrametric trees behave well.
\end{ex}

Consequently, we restrict our attention to a single-class coalescent mixture model on a possibly non-ultrametric species tree, with GTR parameters $Q$, $\pi$. 
We allow population functions to depend on edges in the species tree and, in addition, allow scalar-valued rate functions to be
edge-dependent too.  That is, our model parameters are an arbitrary rooted metric species tree $\sigma$ in units of generations,  
GTR parameters $Q,\pi$, and a collection of population functions $N_e$ and rate functions $\mu_e$ for each edge of $\sigma$. 
We denote this model, which generalizes {\coalModel}+\gtrmu, by $\mathcal N$. We note also that this model is a generalization 
of that of \cite{Roch15}, in that it allows  a GTR substitution model, time-varying population sizes, and scalar-valued mutation rates 
on each edge of the species tree.

\begin{thm}\label{thm:nonultra}
Under the coalescent mixture model $\mathcal N$ on a rooted 
metric species tree, the unrooted tree topology is identifiable from the 
expected pairwise log-det distances.
\end{thm}
\begin{proof}
It suffices to consider 4-taxon species trees, and show the expected log-det distance satisfies the 4-point condition. 
Since pendant edges add the same values to all pairwise distances, 
we can reduce to the case that all terminal branches of the species tree have length 0.  We therefore 
need only consider the two species trees
$((a\tc 0,b\tc 0)\tc x,\, (c\tc 0,d\tc 0)\tc y)$ and $(((a\tc 0,b\tc 0)\tc x,\, c\tc 0)\tc y, \, d\tc 0)$.

For both these trees, the equality in the 4-point condition \eqref{eq:4point} is a consequence of the 
exchangeability of taxa $a$ and $b$ in the model.

To prove the inequality, we begin with the caterpillar tree $(((a\tc 0,b\tc 0)\tc x,c\tc 0)\tc y,d\tc 0)$.
Numbering edge rate functions  by the number of the edge's descendants,
define the rate function
$$\mu(t)=\begin{cases} \mu_2(t)&\text{ if }0\le t<x,\\
\mu_3(t)&\text{ if }x\le t<x+y,\\
\mu_4(t)&\text{ if }x+y\le t,\end{cases}$$ and a population function
 $N(t)$ similarly from the edge population functions.
In the notation of Lemmas \ref{lem:Markov} and \ref{lem:delaycoal},  pairwise pattern frequency arrays are then
\begin{align*}
F_{ab}&=C(Q,\mu,N,0),\\
F_{cd}&=C(Q,\mu,N,x+y)M(\mu,Q,x)^{-1}M(\mu,Q,x+y)^{-1},\\
F_{ac}&=C(Q,\mu, N, x)M(\mu,Q,x)^{-1},\\
F_{bd}&=C(Q,\mu,N,x+y) M(\mu,Q,x+y)^{-1}.
\end{align*}
Using the formula for the log-det distance, 
one sees that the 4-point inequality will follow from
$$\det(C(Q,\mu,N,0))>\det(C(Q,\mu,N,x)).$$ 
But this holds by Lemma \ref{lem:delaycoal} and Lemma \ref{lem:main} (applied to a 1-class mixture).

For the balanced species tree $((a\tc 0,b\tc 0)\tc x,(c\tc 0,d\tc 0)\tc y)$, consider population and rate
functions $N_1$, $\mu_1$, $N_2$, $\mu_2$, 
where those with subscript 1 (respectively 2) are defined piecewise using all edges above the most recent common ancestor of $a,b$ (respectively $c,d$). Then
\begin{align*}
F_{ab}&=C(Q,\mu_1,N_1,0),\\
F_{cd}&=C(Q,\mu_2,N_2,0),\\
F_{ac}=F_{bd}&=C(Q,\mu_1, N_1, x)M(\mu_1,Q,x)^{-1}M(\mu_2,Q,y)\\
&=C(Q,\mu_2, N_2, y)M(\mu_2,Q,y)^{-1}M(\mu_1,Q,x).
\end{align*}
Using the definition of the log-det distance, these expressions show the 4-point inequality will follow from showing
both
$$\det(C(Q,\mu_1,N_1,0))>\det(C(Q,\mu_1,N_1,x))$$ and
$$\det(C(Q,\mu_2,N_2,0))>\det(C(Q,\mu_2,N_2,y)).$$ But these are established by Lemmas \ref{lem:delaycoal} and \ref{lem:main}.
\end{proof}

\section{Discussion}\label{sec:disc}

We have proved that for an ultrametric species tree, using the log-det distance and the 3-point condition, the rooted species tree topology 
is identifiable for very general models of sequence evolution. 
These models include the coalescent process, with varying population sizes, modeling incomplete lineage sorting (ILS) 
and a many-class mixture of base substitution processes with rate variation that is appropriate for genomic sequence data, including 
one with a symmetric form of lineage-specific heterotachy. Although for non-ultrametric species trees one must restrict to a non-mixture model 
with only a single class for the substitution process, even in this case we can establish the identifiability of the unrooted
species topology using the 4-point condition. 

With identifiability proved, it follows easily that inference of the species tree by computing the log-det distance from 
concatenated genomic data, and then applying standard distance methods of tree building or selection yields statistically consistent
topology estimates for these models. This includes using well-known methods such as UPGMA (for ultrametric trees only), 
Neighbor Joining, and Minimum Evolution 
(though this last is usually implemented heuristically, rather than exactly).
These results show both the potential (consistent estimates of topologies can be computed 
quickly for a very general model) and pitfalls (non-ultrametric tree topology estimates are only consistent 
under a more restrictive model) of using log-det distance based tree construction for species tree inference.

\smallskip

Previous works addressing ILS in inference from aligned genomic data have typically assumed a single base substitution 
process with constant population sizes and constant scalar-valued mutation rates on each edge of the tree. In addition 
to their being less plausible biologically,  models assuming  constant population sizes and mutation rates on edges are 
not well-behaved under passing to subsets of taxa, as edges of the induced tree are formed by merging edges 
in the original tree, and thus may not have constant populations and mutation rates.
To eliminate this, 
one can posit a model with a single population size and mutation rate over the entire tree, but that further reduces biological plausibility. 
The modeling assumptions behind a log-det approach do not have these weaknesses.

Moreover, our modeling assumption of ultrametricity of the species tree in generations is one whose reasonableness can often be judged for biological datasets. For instance, for a collection of insects with rigid lifecycles of roughly equal duration, or even for larger organisms known to have similar generation times at present, it may be quite plausible. Although some species tree inference methods which infer individual gene trees can also avoid  
making strong assumptions relating the three time scales inherent to the inference problem (in generations, coalescent units, and  substitution units), they do so at the cost of either introducing poorly-understood inference error in the gene trees, or the computational burden of Bayesian methods. In fact, Bayesian methods often restrict models of  populations size to quite simple families, and little is done to judge their plausibility. Thus all current methods are built on assumptions that may be seriously violated for some data sets.

In addition to ultrametricity, we made two other significant assumptions. 
One is that heterotachy of substitution processes must be either time-dependent, or at worst exhibit the sort of 
``average time-dependence across the genome" captured in Theorem \ref{thm:sym}. In particular, if on species tree 
branches leading to some taxa the rate of mutation increased across the entire genome, then our strongest result is  
Theorem \ref{thm:nonultra} which does not allow a mixture of substitution processes.
The other assumption is that substitution processes are time-reversible. Of course this is routinely invoked by standard methods of gene tree inference, as well as other methods of species tree inference. Nonetheless, we highlight that our arguments using quadratic forms require time-reversibility for every substitution process, and the robustness of the inference scheme to its violation has not been explored.

\smallskip

While the simulation study of \cite{Rusinko17} is not strictly applicable to testing the log-det distance method of species tree inference studied here (as it used a GTR+Gamma distance rather than log-det, with simulation conditions deviating from our model), our own investigations (not shown) indicate similar strong performance of log-det, both on the same simulations and others. However, adequate testing will require much careful thought in designing simulation conditions that are plausible for genomic data. For instance, whether rate matrices for the various classes should be sampled from a  unimodal distribution, or deviations from a molecular clock for individual genes can be modeled by independent scalings on edges, as they were in the simulated data of \cite{Chou2013} and \cite{Bayzid2015} used by \cite{Rusinko17}, may or may not
be acceptable approximations of empirical data. In our view, the results of \cite{Rusinko17} 
underscore that in current simulations of genomic data it is not clear whether an  `easy' or `difficult' region of parameter space is being investigated. While the anomolous gene trees concept \cite{DegRos2006,kubatko2007inconsistency} has illuminated how expected gene tree topologies may make species tree inference difficult, the impact of variation in substitution process across the genome is largely unexplored.

Nonetheless, the generality of the model to which log-det distance based inference can be consistently applied, combined with its simplicity of implementation, speed on large data sets,
and preliminary indications of good performance in simulation, suggest this method should provide a basic benchmark to which other species tree inference methods should be compared.

\section*{Acknowledgements}
Work on this project was supported by the National Institutes of Health grant R01 GM117590, awarded under the  Joint DMS/NIGMS Initiative to Support Research at the Interface of the Biological and Mathematical Sciences to ESA and JAR, and by the Mathematical Biosciences Institute and the National Science Foundation under grant DMS-1440386 for CL. 

\bibliographystyle{plain}
\bibliography{logdet_arxiv}

\end{document}